\documentclass[11pt, letterpaper]{article}
\usepackage[utf8]{inputenc}
\usepackage[UKenglish]{babel} 

\usepackage{authblk} 
\usepackage{fullpage}
\usepackage{csquotes}    

\usepackage{fullpage} 

\usepackage[backend=biber, uniquelist=false, sorting=nyt, style=alphabetic]{biblatex}

\usepackage{graphicx} 

\usepackage{todonotes} 

\usepackage[ampersand]{easylist}

\usepackage{amsthm}
\usepackage{amssymb}
\usepackage{amsmath}

\usepackage{algorithm}
\usepackage[noend]{algpseudocode}

\usepackage{hyperref}
\hypersetup{
    colorlinks=true,
    linkcolor=red,
    filecolor=magenta,   citecolor=orange,
    urlcolor=blue,
}

\theoremstyle{plain}
\newtheorem{theorem}{Theorem}

\newtheorem{definition}[theorem]{Definition} 
\newtheorem{lemma}[theorem]{Lemma}
\newtheorem{claim}[theorem]{Claim}

\newtheorem*{remark*}{Remark}

\newtheorem{manualtheoreminner}{Theorem}
\newenvironment{manualtheorem}[1]{%
  \renewcommand\themanualtheoreminner{#1}%
  \manualtheoreminner
}{\endmanualtheoreminner}

\usepackage{thmtools}
\usepackage{thm-restate}

\title{On Two-Pass Streaming Algorithms for Maximum Bipartite Matching}
\author{Christian Konrad}
\author{Kheeran K. Naidu}
\affil{Department of Computer Science, University of Bristol, UK \\ \texttt{\{christian.konrad,kheeran.naidu\}@bristol.ac.uk}}
\date{}

\begin{document}
\maketitle
\thispagestyle{empty}
\begin{abstract}
We study two-pass streaming algorithms for \textsf{Maximum Bipartite Matching} (\textsf{MBM}). All known two-pass streaming algorithms for \textsf{MBM} operate in a similar fashion: They compute a maximal matching in the first pass and find 3-augmenting paths in the second in order to augment the matching found in the first pass. Our aim is to explore the limitations of this approach and to determine whether current techniques can be used to further improve the state-of-the-art algorithms. We give the following results:

We show that every two-pass streaming algorithm that solely computes a maximal matching in the first pass and outputs a $(2/3+\epsilon)$-approximation requires $n^{1+\Omega(\frac{1}{\log \log n})}$ space, for every $\epsilon > 0$, where $n$ is the number of vertices of the input graph. This result is obtained by extending the Ruzsa-Szemer\'{e}di graph construction of [Goel et al., SODA'12] so as to ensure that the resulting graph has a close to perfect matching, the key property needed in our construction. This result may be of independent interest.

Furthermore, we combine the two main techniques, i.e., subsampling followed by the \textsc{Greedy} matching algorithm [Konrad, MFCS'18] which gives a $2-\sqrt{2} \approx 0.5857$-approximation, and the computation of \emph{degree-bounded semi-matchings} [Esfandiari et al., ICDMW'16][Kale and Tirodkar, APPROX'17] which gives a $\frac{1}{2} + \frac{1}{12} \approx 0.5833$-approximation, and obtain a meta-algorithm that yields Konrad's and Esfandiari et al.'s algorithms as special cases. This unifies two strands of research. By optimizing parameters, we discover that Konrad's algorithm is optimal for the implied class of algorithms and, perhaps surprisingly, that there is a second optimal algorithm. We show that the analysis of our meta-algorithm is best possible.
Our results imply that further improvements, if possible, require new techniques.
\end{abstract}

\pagebreak

\section{Introduction} \label{sec:intro}  
In the \emph{semi-streaming model} for processing large graphs, an $n$-vertex graph is presented to an algorithm as a sequence of its edges in arbitrary order. The algorithm makes one or few passes over the input stream and maintains a memory of size $O(n \mathop{\mathrm{polylog}} n)$. 
  
The semi-streaming model has been extensively studied since its introduction by Feigenbaum et al. in 2004 \cite{fkmsz04}, and various graph problems, including matchings, independent sets, spanning trees, graph sparsification, subgraph detection, and others are known to admit semi-streaming algorithms (see \cite{m14} for an excellent survey). Among these problems, the \textsf{Maximum Matching} problem and, in particular, its bipartite version, the \textsf{Maximum Bipartite Matching} (\textsf{MBM}) problem, have received the most attention (see, for example, \cite{fkmsz04,m05,ag11,kmm12,gkk12, kapralov2013better,ehm16, kt17,k18,gkms19,b20,fhmrr20,ab21,alt21,k21}).

In this paper, we focus on \textsf{MBM}. The currently best one-pass semi-streaming algorithm for \textsf{MBM} is the \textsc{Greedy} matching algorithm (depicted in Algorithm~\ref{alg:greedy}). \textsc{Greedy} processes the edges of a graph in arbitrary order and inserts the current edge into an initially empty matching if possible. It produces a maximal matching, which is known to be at least half the size of a maximum matching, and constitutes a $\frac{1}{2}$-approximation semi-streaming algorithm for \textsf{MBM}. It is a long-standing open question whether \textsc{Greedy} is optimal for the class of semi-streaming algorithms or whether an improved approximation ratio is possible. Progress has been made on the lower bound side (\cite{gkk12,kapralov2013better,k21}), ruling out semi-streaming algorithms with approximation ratio better than $\frac{1 }{1 + \ln 2} \approx 0.5906$ \cite{k21}. 

Konrad et al. \cite{kmm12} were the first to show that an approximation ratio better than $\frac{1}{2}$ can be achieved if two passes over the input are allowed, and further successive improvements \cite{kt17,ehm16,k18} led to a two-pass semi-streaming algorithm with an approximation factor of $2 - \sqrt{2} \approx 0.58578$ \cite{k18} (see Table~\ref{tab:overview} for an overview of two-pass algorithms for \textsf{MBM}).

\begin{table}[ht]
\centering
 \begin{tabular}{|rlc|}
  \hline   
  Approximation Factor & Reference & Comment \\
  \hline 
  $\frac{1}{2}+0.019$ & Konrad et al. \cite{kmm12} & randomized \\
  $\frac{1}{2} + \frac{1}{16} = 0.5625$ & Kale and Tirodkar \cite{kt17} & deterministic \\
  $\frac{1}{2} + \frac{1}{12} \approx 0.5833$ & Esfandiari et al. \cite{ehm16} & deterministic \\
  $2 - \sqrt{2} \approx 0.5857$ & Konrad \cite{k18} & randomized \\
  \hline
 \end{tabular}
\caption{Two-pass semi-streaming algorithms for \textsf{Maximum Bipartite Matching}. \label{tab:overview}} 
\end{table}

All known two-pass streaming algorithms proceed in a similar fashion. In the first pass, they run \textsc{Greedy} in order to compute a maximal matching $M$. In the second pass, they pursue different strategies to compute additional edges $F$ that allow them to increase the size of $M$. Two techniques for computing the edge set $F$ have been used:

\begin{figure}[t]
\centering
\begin{minipage}{0.48\textwidth}
 \begin{algorithm}[H]
 \textbf{Input: } Graph $G=(A, B, E)$
  \begin{algorithmic}[1]
\State $M \gets \varnothing$
\State \textbf{for each} edge $e \in E$ (arbitrary order)
\State $\quad$ \textbf{if} $M \cup \{e \}$ is a matching
\State $\quad$ $\quad$ $M \gets M \cup \{ e \}$
\State \textbf{return} $M$
  \end{algorithmic}
  \caption{\textsc{Greedy} Matching. \label{alg:greedy}}
 \end{algorithm}
\end{minipage} \hfill 
\begin{minipage}{0.48\textwidth}
 \begin{algorithm}[H]
 \textbf{Input: } Graph $G=(A, B, E)$, integer $d$
  \begin{algorithmic}[1]
\State $S \gets \varnothing$
\State \textbf{for each} edge $ab \in E$ (arbitrary order)
\State $\quad$ \textbf{if} $\deg_S(a) = 0$ \textbf{and} $\deg_S(b) < d$ 
\State $\quad$ $\quad$ $S \gets S \cup \{ ab \}$
\State \textbf{return} $S$
  \end{algorithmic}
  \caption{$\textsc{Greedy}_d$ Semi-Matching. \label{alg:greedyd}}
 \end{algorithm}
\end{minipage}
\end{figure}

\begin{enumerate}
 \item \textbf{Subsampling and \textsc{Greedy}} \cite{k18} (see also \cite{kmm12}): Given a bipartite graph $G = (A, B, E)$ and a first-pass maximal matching $M$, they first subsample the edges $M$ with probability $p$ and obtain a matching $M' \subseteq M$. Then, in the second pass, they compute \textsc{Greedy} matchings $M_L$ and $M_R$ on subgraphs $G_L = G[A(M') \cup \overline{B(M)}]$ and $G_R = G[\overline{A(M)} \cup B(M')]$, respectively, where $A(M')$ are the matched $A$-vertices in $M'$, $\overline{B(M)}$ are the unmatched $B$ vertices, and $B(M')$ and $\overline{A(M)}$ are defined similarly. It can be seen that if  $M$ is relatively small, then $M' \cup M_L \cup M_R$ contains many disjoint $3$-augmenting paths. Setting $p=\sqrt{2}-1$ yields the approximation factor $2-\sqrt{2}$.
 
 \item \textbf{Semi-matchings and $\textsc{Greedy}_d$} \cite{kt17,ehm16}: Given a bipartite graph $G = (A, B, E)$ and a first-pass maximal matching $M$, the second pass consists of finding \textit{degree-$d$-constrained semi-matchings} $S_L$ and $S_R$ on subgraphs $G_L = G[A(M) \cup \overline{B(M)}]$ and $G_R = G[\overline{A(M)} \cup B(M)]$, respectively, using the algorithm $\textsc{Greedy}_d$ (as depicted in Algorithm~\ref{alg:greedyd}). A degree-$d$-constrained semi-matching in a bipartite graph is a subset of edges $S \subseteq E$ such that $\deg_S(a) \le 1$ and $\deg_S(b) \le d$, for every $a \in A$ and $b \in B$ or vice versa\footnote{The usual definition of a semi-matching requires $\deg_S(a) = 1$, for every $a \in A$ (e.g. \cite{fln14,kr16}). This property is not required here, and, for ease of notation, we stick to this term.}. Similar to the method above, it can be seen that if the matching $M$ is relatively small, $M \cup S_L \cup S_R$ contains many disjoint $3$-augmenting paths. The setting $d=3$ yields the approximation factor $\frac{1}{2} + \frac{1}{12}$. 
\end{enumerate}

\subparagraph*{Our Results.} 
In this paper, we explore the limitations of this approach and investigate whether current techniques can be used to further improve the state-of-the-art.

Our first result is a limitation result on the approximation factor achievable by algorithms that follow the scheme described above:

\begin{theorem}[simplified] \label{thm:lb} 
 Every two-pass semi-streaming algorithm for \textsf{MBM} that solely runs \textsc{Greedy} in the first pass has an approximation factor of at most $\frac{2}{3}$.
\end{theorem}
Our result builds upon a result by Goel et al. \cite{gkk12} who proved that the lower bound of Theorem~\ref{thm:lb} applies to one-pass streaming algorithms. Their construction relies on the existence of dense \emph{Ruzsa-Szemer\'{e}di} graphs with large induced matchings, i.e., bipartite $2n$-vertex graphs $G=(A, B, E)$ with $|A| = |B|= n$ whose edge sets can be partitioned into disjoint induced matchings such that each matching is of size at least $(\frac{1}{2}- \delta)n$, for some small $\delta$. Our construction requires similarly dense RS graphs with equally large matchings, however, in addition to these properties, our RS graphs must contain a \emph{near-perfect} matching, i.e., a matching that matches all but a small constant fraction of the vertices. To this end, we augment the RS graph construction by Goel et al.: We show that, for each induced matching $M$ in Goel et al.'s construction, we can add a matching $M'$ to the construction without violating the induced matching property such that $M \cup M'$ forms a near-perfect matching. We believe this result may be of independent interest.

Next, we combine the subsampling and semi-matching techniques and give a meta-algorithm that yields Konrad's and Esfandiari et al.'s algorithms as special cases, thereby unifying two strands of research. 
Our meta-algorithm is parameterised by a sampling probability $0 < p \le 1$ and an integral degree bound $d \ge 1$. First, as in the subsampling technique, the edges of the first-pass matching $M$ are subsampled independently with probability $p$, which yields a subset $M' \subseteq M$. Next, as in the semi-matching technique, incomplete semi-matchings $S_L$ and $S_R$ with degree bounds $d$ are computed, however, now in the subgraphs $G_L' = G[A(M') \cup \overline{B(M)}]$ and $G_R' = G[\overline{A(M)} \cup B(M')]$. The algorithm then outputs the largest matching among the edges $M \cup S_L \cup S_R$.

As our second result, we establish the approximation factor of our meta-algorithm:

\begin{theorem}[simplified]\label{thm:ub}
 Combining the subsampling and semi-matching techniques yields a two-pass semi-streaming algorithm for \textsf{MBM} with approximation factor 
 \begin{align*}
  \begin{cases}
   \frac{1}{2} + (\frac{1}{d+p} - \frac{1}{2d}) \cdot p, & \mbox{ if } p \le d(\sqrt{2}-1) \\
   \frac{1}{2} + \frac{d-p}{6d+2p}, & \ \mbox{ otherwise} \ ,
  \end{cases}
 \end{align*}
 (ignoring lower order terms) that succeeds with high probability.
\end{theorem}

\begin{figure}[t]
\centering
\includegraphics[width=0.95\textwidth]{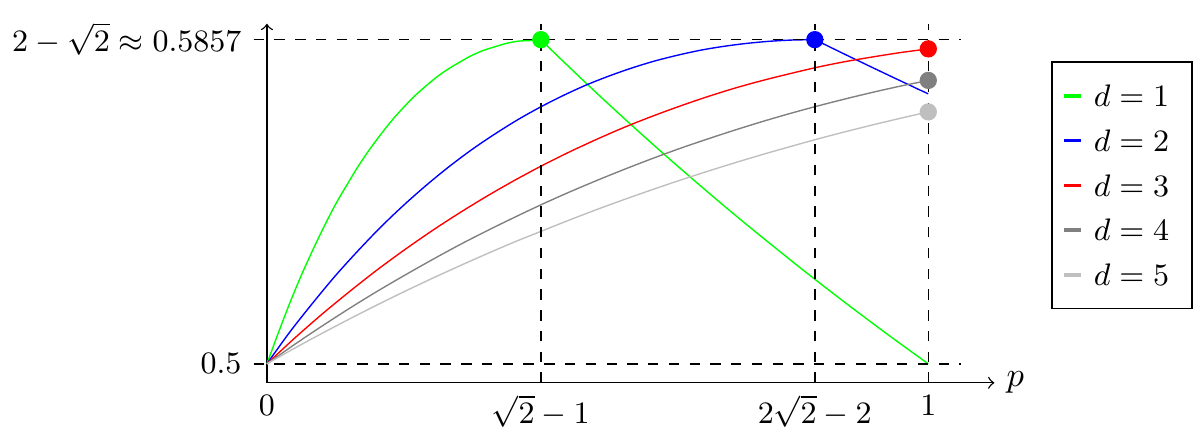}
\caption{Approximation factors for different settings of $d$.} \label{fig:meta_alg}
\end{figure}

Interestingly, two parameter settings maximize the approximation factor in Theorem~\ref{thm:ub}, achieving the ratio $2-\sqrt{2}$ (see Figure~\ref{fig:meta_alg}). This is achieved by setting $d=1$ and $p=\sqrt{2}-1$ which recovers Konrad's algorithm, and by setting $d=2$ and $p= 2\sqrt{2}-2$ which gives a new algorithm. The setting $d=3$ and $p=1$ yields the slightly weaker bound $\frac{1}{2} + \frac{1}{12} \approx 0.5833$ and recovers Esfandiari et al.'s algorithm. 

We also show that the analysis of our meta-algorithm is tight, by giving instances on which our meta-algorithm does not perform better than the claimed bound (\textbf{Theorem~\ref{thm:opt_analysis}}).   

\subparagraph*{Discussion.} 
Our results demonstrate that new techniques are needed in order to improve on the $(2-\sqrt{2})$ approximation factor.
However, one may wonder whether $2-\sqrt{2}$ is the best approximation ratio achievable by the class of two-pass matching algorithms that solely computes a maximal matching in the first pass. As pointed out by Kapralov \cite{k21}, his techniques for establishing the $\frac{1}{1+\ln 2}$ lower bound for one-pass algorithms can probably also be applied to a construction by Huang et al. \cite{hpttwz19}, which would then show that $2-\sqrt{2}$ is the best approximation factor achievable by one-pass semi-streaming algorithms for \textsf{MBM}. It is unclear whether a first-pass \textsc{Greedy} matching could be embedded in the resulting construction without affecting its hardness, however, if possible, this would render Konrad's algorithm optimal for the considered class of two-pass streaming algorithms.

\subparagraph*{Independent and Concurrent Work.}
Independently and concurrently to our work, Assadi \cite{a21} gave a limitation result on the approximation factor achievable by {\em arbitrary} two-pass semi-streaming algorithms for \textsf{MBM}, i.e., algorithms with no restrictions on how they operate in the first pass. From the lower bound perspective, this setting is substantially harder to work with than the setting considered in this paper, where we assume that algorithms run \textsc{Greedy} in the first pass. Assadi showed that no (arbitrary) two-pass semi-streaming algorithm for \textsf{MBM} has an approximation factor better than $\left(1 - \Omega(\frac{\log \text{RS}(n)}{\log n})\right)$, where $\text{RS}(n)$ denotes the maximum number of disjoint induced matchings of size $\Omega(n)$ in any $n$-vertex graph. Determining $\text{RS}(n)$ is a challenging open problem in combinatorics, and, currently,  the best upper and lower bounds are still very far apart from each other:  $n^{\Omega(\frac{1}{\log \log n})} \le \text{RS}(n) \le n^{1-o(1)}$ \cite{flnrrs02, fhs17}. In the best case scenario, i.e., if $\text{RS}(n)$ was indeed as large as $n^{1-o(1)}$, their result would imply that no two-pass semi-streaming algorithm can achieve a better than $0.98$-approximation, and as long as $\text{RS}(n) = n^{\Omega(1)}$, their result would rule out small constant factor approximations.

\subparagraph*{Further Related Work.}
Besides two passes over the input, improvements over the \textsc{Greedy} algorithm can also be obtained under the assumption that the input stream is in random order. Assadi and Behnezhad \cite{ab21} recently showed that an approximation factor of $\frac{2}{3}+\epsilon$ can be obtained, for some fixed small but constant $\epsilon > 0$, building on Bernstein's breakthrough result \cite{b20}, and improving on previous algorithms \cite{b20,fhmrr20,k18,kmm12}. 
In insertion-deletion streams, where previously inserted edges may be deleted again, space $\tilde{\Theta}(n^{2-3\epsilon})$ is necessary \cite{dk20} and sufficient \cite{akly16,ccehmmv16} for computing a $n^\epsilon$-approximation (see also \cite{k15}).

\subparagraph*{Outline.} 
We first give notation and definitions in Section~\ref{sec:prelim}. Subsequently, we show in Section~\ref{sec:lb} that every two-pass semi-streaming algorithm that solely runs \textsc{Greedy} in the first pass cannot have an approximation ratio of $\frac{2}{3} + \epsilon$, for any $\epsilon > 0$. Our main algorithmic result, i.e., the combination of subsampling and $\textsc{Greedy}_d$, is presented in Section~\ref{sec:ub}. Finally, we conclude in Section~\ref{sec:conclusion}.

\section{Preliminaries}\label{sec:prelim}

Let $G =(A,B,E)$ be a bipartite graph with $V = A \cup B$ and $|V| = n$. For $F \subseteq E$ and $v \in V$, we write $\deg_F(v)$ to denote the degree of vertex $v$ in subgraph $(A, B, F)$. For any $U \subseteq V$ and $F \subseteq E$, $U(F)$ denotes the set of vertices in $U$ which are the endpoints of edges in $F$, and we denote its complement by $\overline{U(F)} = U \setminus U(F)$. For a subset of vertices $U \subseteq V$, we write $G[U]$ for the subgraph of $G$ induced by $U$.
For any edges $e,f \in E$, $e$ is \emph{incident} to $f$ if they share an endpoint. We say that $e$ and $f$ are \emph{vertex-disjoint} if $e$ is not incident to $f$. 
Lastly, for any two sets $X$ and $Y$, we define $X \oplus Y := (X \setminus Y) \cup (Y \setminus X)$ as their symmetric difference. 

A \emph{matching} in $G$ is a subset $M \subseteq E$ of vertex-disjoint edges. It is \emph{maximal} if every $e \in E \setminus M$ is incident to an edge in $M$.  We denote by $\mu(G)$ the \emph{matching number} of $G$, i.e., the cardinality of a largest matching. A \emph{maximum matching} is one of size $\mu(G)$. Additionally, $M$ is called an \emph{induced matching} if the edge set of the subgraph of $G$ induced by $V(M)$ is exactly $M$. 

\subparagraph*{Wald's Equation.}
We require the following well-known version of \textit{Wald's Equation}:

\begin{restatable}{lemma}{waldseq} \label{lem:walds_eq}
Let $X_1, X_2, \dots$ be a sequence of non-negative random variables with $\mathbb{E}[X_i] \leq \tau$, for all $i \leq T$, and let $T$ be a random stopping time for the sequence with $\mathbb{E}[T] < \infty$. Then:
\begin{align*}
\mathbb{E}[\sum_{i=1}^T X_i] \leq \tau \cdot \mathbb{E}[T] \ .
\end{align*}
\end{restatable}

\section{Lower Bound} \label{sec:lb}
We now prove that every two-pass streaming algorithm for \textsf{MBM} with approximation factor $\frac{2}{3}+\epsilon$, for any $\epsilon > 0$, that solely runs \textsc{Greedy} in the first pass requires space $n^{1+\Omega(\frac{1}{\log \log n})}$. To this end, we adapt the lower bound by Goel et al. \cite{gkk12}, which we discuss first.

\subsection{Goel et al.'s Lower Bound for One-pass Algorithms}
Goel et al.'s lower bound is proved in the \emph{one-way two-party communication framework}. Two parties, denoted Alice and Bob, each hold subsets $E_1$ and $E_2$, respectively, of the input graph's edges. Alice sends a single message to Bob who, upon receipt, outputs a large matching. Goel et al. showed that there is a distribution $\lambda$ over input graphs so that every deterministic communication protocol with constant distributional error over $\lambda$ and approximation factor $\frac{2}{3}+\epsilon$, for any $\epsilon > 0$, requires a message of length $n^{1+\Omega(\frac{1}{\log \log n})}$. 
A similar result then applies for randomized constant error protocols by Yao's Lemma~\cite{yao1977probabilistic}, and the well-known connection between streaming algorithms and one-way communication protocols allows us to translate this lower bound to a lower bound on the space requirements of constant error one-pass streaming algorithms.

Goel et al.'s construction is based on the existence of a dense 
Ruzsa-Szemer\'{e}di graph:
\begin{definition}[Ruzsa-Szemer\'{e}di Graph] \label{defn:RS_graph}
 A bipartite graph $G=(A, B, E)$ is an \emph{$(r,t)$-Ruzsa-Szemer\'{e}di} graph (RS graph in short) if the edge set $E$ can be partitioned into $t$ disjoint matchings $M_1, M_2, \dots, M_t$ such that, for every $i$, (1) $|M_i| \ge r$; and (2) $M_i$ is an \emph{induced matching} in $G$.
\end{definition}
They give a construction for a family of $((\frac{1}{2} - \delta)n, n^{\Omega(\frac{1}{\log \log n})})$-RS graphs, for any small constant $\delta > 0$, on $2n$ vertices (with $|A| = |B| = n$) that we will extend further below.

\begin{figure}[t]
\centering
\fbox{
\begin{minipage}{0.95 \textwidth}
\begin{enumerate}
\item Let $G^{RS}=(A, B, E)$ be an $(r,t)$-RS graph with $|A| = |B| = N$ and $r = (\frac{1}{2} - \delta) \cdot N$, for some $\delta > 0$, and $t = N^{\Omega(\frac{1}{\log \log N})}$. 
\item For every $i \in [t]$, let $\widehat{M_i} \subseteq M_i$ be a uniform random subset of size $(\frac{1}{2} - 2\delta) \cdot N$ and let $E_1 = \cup_{i =1}^t \widehat{M_i}$.
\item Let $X$ and $Y$ each be disjoint sets of $(\frac{1}{2} + \delta) \cdot N$ vertices, which are also disjoint from $A \cup B$. Choose uniformly at random a special index $s \in [t]$.
\item Let $M_X^*$ and $M_Y^*$ be arbitrary perfect matchings between $X$ and $\overline{B(M_s)}$, and $Y$ and $\overline{A(M_s)}$, respectively. Then, let $E_2 = M_X^* \cup M_Y^*$.
\item Finally, $G = (A \cup X, B \cup Y, E_1 \cup E_2)$ which has $n = (3 + 2\delta) \cdot N$ vertices.
\end{enumerate}

Alice is given edges $E_1$ and Bob is given edges $E_2$.
\end{minipage}
}
\caption{Hard input distribution $\lambda$. \label{fig:lambda}}
\end{figure}

Their hard input distribution $\lambda$ for the two-party communication setting is displayed in Figure~\ref{fig:lambda}. Observe that the graphs $G \sim \lambda$ are such that $\mu(G) \ge \frac{3}{2} N$ since the matching $M_X^* \cup M_Y^* \cup \widehat{M_s}$ is of this size.

Goel et al. prove the following hardness result: 
\begin{theorem}\label{thm:goel-hardness}
 For any small $\epsilon > 0$, every deterministic $(\frac{2}{3} + \epsilon)$-approximation one-way two-party communication protocol with constant distributional error over $\lambda$ requires a message of size $n^{1+\Omega(\frac{1}{\log \log n})}$, where $n$ is the number of vertices in the input graph.
\end{theorem}

\subsection{Our Lower Bound Construction}
In the following, we extend Goel et al.'s lower bound to the two-pass situation where a $\textsc{Greedy}$ matching is computed in the first pass. To this end, we need to augment Alice and Bob's inputs, as defined by distribution $\lambda$, by a maximal matching $M$ in the input graph $G \sim \lambda$, which then results in a distribution $\lambda^+$. Observe that if we place the edges of $M$ at the beginning of the input stream, then running \textsc{Greedy} in the first pass recovers exactly the matching $M$. Hence, when abstracting the second pass as a two-party communication problem, both Alice and Bob already know the matching $M$. Our main argument then is as follows: We will show that any two-party protocol under distribution $\lambda^+$ can also be used for solving the distribution $\lambda$ with the same distributional error, message size, and similar approximation factor. The hardness of Theorem~\ref{thm:goel-hardness}, therefore, carries over.

\subsubsection{Ruzsa-Szemer\'{e}di Graphs with Near-Perfect Matchings}
Adding a maximal matching $M$ to Alice's and Bob's input requires care since we need to ensure that the hardness of the construction is preserved. Our construction requires that the underlying RS graph contains a near-perfect matching, which is a property that is not guaranteed by Goel et al.'s RS graph construction. 

We therefore augment Goel et al.'s construction by  complementing every induced matching, $M_i$, with a vertex-disjoint counterpart, $M_i'$, without destroying the RS graph properties. Then, since $M_i$ and $M_i'$ are vertex-disjoint, $M_i \cup M_i'$ constitutes a matching, and, since both $M_i$ and $M_i'$ each already match nearly half of the vertices, $M_i \cup M_i'$ constitutes a near-perfect matching in our family of RS graphs. 

We will now present Goel et al.'s RS graph construction and then discuss how the additional matchings $M_i'$ can be added to the construction.

\paragraph*{Goel et al.'s Ruzsa-Szemer\'{e}di Graph Construction} \mbox{} \vspace{5pt} \\ \noindent
For an integer $m$, let $X = Y = [m^2]^m$ be the vertex sets of a bipartite graph, and let $N = |X| = |Y| =  m^{2m}$ denote their cardinalities. Every induced matching $M_I$ of Goel et al.'s RS graph construction is indexed by a subset of coordinates $I \subset [m]$ of size $\frac{\delta m}{6}$, for some small $ \delta > 0$. Then, the edges $M_I$ are defined by means of a colouring of the vertices $X$ and $Y$ (which depends on $I$), that we discuss first.

\subparagraph*{Colouring the Vertex Sets.} Let $w = \frac{(2 + \delta)m}{3}$. Then, define a partition of the natural numbers into groups of size $w$ such that, for all $k \in \mathbb{N}_0$,
\begin{align*}
R_k &= \left[ kw, kw + \frac{m}{3} \right) & \textrm{where } |R_k| = \frac{m}{3}, \\
W_k &=	\left[ kw + \frac{m}{3}, kw + \frac{m}{3} + \frac{\delta m}{6} \right) & \textrm{where }|W_k| = \frac{\delta m}{6}, \\
B_k &= \left[ kw + \frac{m}{3} + \frac{\delta m}{6},  kw + \frac{2m}{3} + \frac{\delta m}{6} \right) & \textrm{where } |B_k| = \frac{m}{3}, \\
W_k' &=\left[ kw + \frac{2m}{3} + \frac{\delta m}{6}, (k+1)w \right)  & \textrm{where } |W_k'| = \frac{\delta m}{6}.
\end{align*}
See Figure 3 for an illustration.

\begin{figure}[h]
\centering
\includegraphics[width=0.9\textwidth]{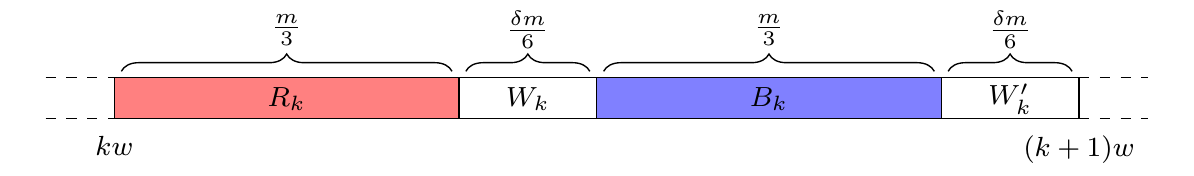}
\caption{One group of the partitioned number line of natural numbers.} \label{fig:nat_partition}
\end{figure}

Given $I$, let $L_s = \{ \vec{x} \in [m^2]^m : \sum_{i \in I} x_i = s \}$ represent a layer of vectors in $[m^2]^m$ where the $1$-norm of their subvectors\footnote{A subvector in this context is the result of a trivial mapping of the vector to a lower dimensional subspace.} w.r.t. $I$ is $s$, for all $s \in \mathbb{N}_0$. Next, colour the vectors in each $L_s$ either red if $s \in R_k$, blue if $s \in B_k$, or white if $s \in W_k \cup W_k'$, for some $k \in \mathbb{N}_0$. Doing this gives the following coloured strips for any $k \in \mathbb{N}_0$ (see Figure~\ref{fig:induced_mat}):
\[R(k) = \bigcup_{\forall s \in R_k} L_s, \quad W(k) = \bigcup_{\forall s \in W_k} L_s, \quad B(k) = \bigcup_{\forall s \in B_k} L_s \quad \textrm{and} \quad W'(k) = \bigcup_{\forall s \in W'_k} L_s.\] Next, these strips are grouped together by colour, as follows:
\[R = \bigcup_{\forall k \in \mathbb{N}_0} R(k), \quad W = \bigcup_{\forall k \in \mathbb{N}_0} W(k), \quad B = \bigcup_{\forall k \in \mathbb{N}_0} B(k) \quad \textrm{and} \quad W' = \bigcup_{\forall k \in \mathbb{N}_0} W'(k).\]

We now define the colours of the vertices $X$ and $Y$ as follows: A vertex $\vec{z} \in X \cup Y$ is coloured red if $\vec{z} \in R$, blue if $\vec{z} \in B$, and white if $\vec{z} \in W \cup W'$. Let $R^X = R \cap X$ and define $B^X, W^X, W'^X, R^Y, B^Y, W^Y, W'^Y$ similarly. 

\begin{figure}[h]
\centering
\includegraphics[width=\textwidth]{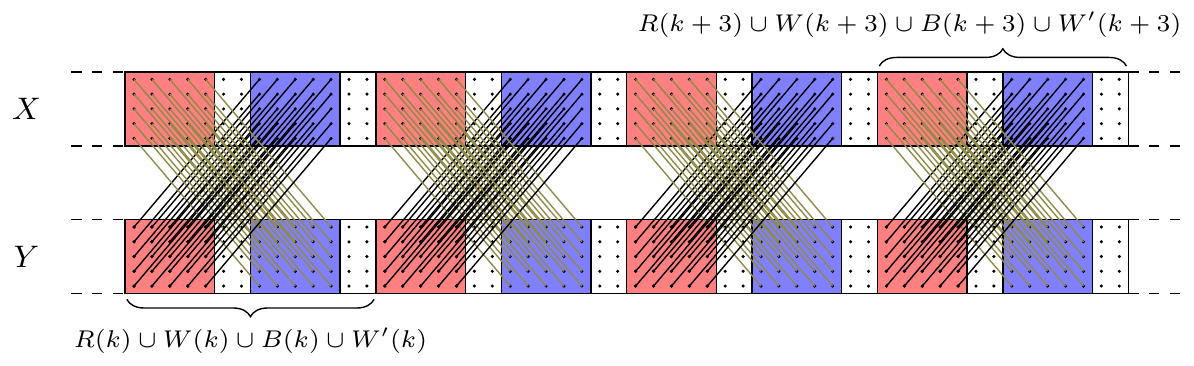}
\caption{Illustration of the vertex colouring and induced matchings for a fixed $I$. The black edges are $M_I$ and the gold ones are $M_I'$.} \label{fig:induced_mat}
\end{figure}

\subparagraph*{Definition of the Induced Matchings.}
Goel et al. construct the edges of the induced matching $M_I$ by pairing every blue vertex $\vec{b} \in B^X$ with each coordinate greater than $\frac{2}{\delta} + 1$ to a red vertex $\vec{r} \in R^Y$, such that 
$\vec{r} = \vec{b} - (\frac{2}{\delta} + 1) \cdot \vec{1}_I$, 
where $\vec{1}_I$ is the characteristic vector of set $I$. See Figure~\ref{fig:induced_mat} for an illustration.

Goel et al. show that $M_I$ is large, i.e., $|M_I| \ge (\frac{1}{2} - \delta) \cdot N - o(N)$. 
Observe that any two distinct indexing sets $I$ and $J$ produce their own vertex colourings and matchings $M_I$ and $M_J$. They prove that, 
as long as the index sets $I$ and $J$ have a sufficiently small intersection (at most $(\frac{5 \delta}{12})( \frac{\delta m}{6})$), $M_I$ and $M_J$ are induced matchings w.r.t. to each other. Hence, they show the existence of a large family $\mathcal{T}$, with $|\mathcal{T}| = N^{\Omega(\frac{1}{\log \log N})}$, of subsets $I \subset [m]$ whose pairwise intersections are of size at most $(\frac{5 \delta}{12})( \frac{\delta m}{6})$. Then, the matchings of the RS graph are identified as the matchings $M_I$, for every $I \in \mathcal{T}$.

\paragraph*{Extending Goel et al.'s Construction} \mbox{} \vspace{5pt} \\ \noindent
For every indexing set $I \in \mathcal{T}$ and respective matching $M_I$ of Goel et al.'s construction, we symmetrically construct an additional matching $M_I'$ by pairing every blue vertex in $Y$ (instead of $X$), $\vec{b} \in B^Y$, with each coordinate greater than $\frac{2}{\delta} + 1$, to a red vertex in $X$, $\vec{r} \in R^X$, such that 
$\vec{r} = \vec{b} - (\frac{2}{\delta} + 1) \cdot \vec{1}_I$. See Figure~\ref{fig:induced_mat} for an illustration.

We immediately see that, by virtue of being symmetrical, $|M_I'| = |M_I| (\ge (\frac{1}{2} - \delta) \cdot N -o(N))$. 

Furthermore, by construction, $M_I'$ and $M_I$ are vertex-disjoint matchings, hence $M_I \cup M_I'$ is a matching, and, taking their respective sizes into account, $M_I \cup M_I'$ is a near-perfect matching as required. Since, for any distinct $I,J \in \mathcal{T}$, $M_I$ and $M_J$ are induced matchings w.r.t. each other, the symmetrical nature of our additional matchings implies the same for $M_I'$ and $M_J'$. However, showing that $M_I$ and $M_J'$ are induced with respect to each other is not immediately clear. Fortunately, Goel et al.'s proof already implicitly shows this, and, for completeness, we reproduce the decisive argument:

\begin{restatable}{lemma}{inducedproof} \label{lem:inducedproof}
Given two distinct sets of indices $I$ and $J$ such that $|I \cap J| \leq (\frac{5 \delta}{12})( \frac{\delta m}{6})$, no edge in $M_{I}$ is induced by $M_{J}'$, for any small enough $\delta > 0$.
\end{restatable}
\begin{proof}
Let $\vec{b} \in B^X$ be matched to $\vec{r} \in R^Y$ by $M_{I}$, i.e., $\vec{b} - \vec{r} = (\frac{2}{\delta} + 1)\cdot \vec{1}_{I}$. If the edge $(\vec{b}, \vec{r})$ is induced by $M_{J}'$, then one endpoint is coloured blue and the other red in the colouring of $X$ and $Y$ with respect to $J$. Hence, $\vec{b}$ and $\vec{r}$ are separated by a single white strip (see Figure~\ref{fig:induced_mat}) and
\begin{align} \label{eq:vec_strip_separate}
	|\sum_{j \in J} (\vec{b} - \vec{r})_j| \geq \frac{\delta m}{6} \ .
\end{align}
On the other hand, 
\begin{align*}
	|\sum_{j \in J} (\vec{b} - \vec{r})_j|  & = |\sum_{j \in J} ((\frac{2}{\delta} + 1)\cdot \vec{1}_{I})_j| = (\frac{2}{\delta} + 1) \cdot |I \cap J| \leq (\frac{5}{6} + \frac{5 \delta}{12})(\frac{\delta m}{6}) \ ,
\end{align*}
which contradicts Equation~\ref{eq:vec_strip_separate}  for small enough $\delta$.
\end{proof}

We thus obtain the following theorem:

\begin{theorem} \label{thm:RSgraphdense+perfect}
For any small enough constant $\delta > 0$, there exists a family of bipartite $(r,t)$-Ruzsa-Szemer\'{e}di graphs where $|A| = |B| = N$, $r = (\frac{1}{2} - \delta) \cdot N$, and $t = N^{\Omega(\frac{1}{\log \log N})}$ such that there are $N^{\Omega(\frac{1}{\log \log N})}$ disjoint near-perfect matchings each of size exactly $(1 - 2\delta) \cdot N$.
\end{theorem}

\subsubsection{Lower Bound Proof}
Equipped with RS graphs with near-perfect matchings and input distribution $\lambda$, we now define our hard input distribution $\lambda^+$, see Figure~\ref{fig:lambda+}. 

\begin{figure}[t]
\centering
\fbox{
\begin{minipage}{0.95\textwidth}
\begin{enumerate}
 \item Let $G^{RS}$ be an RS graph as in Theorem~\ref{thm:RSgraphdense+perfect}. Fix some induced matching $M_i$ and let $M_i \cup M_i'$ be its near-perfect matching of size $(1-2\delta) \cdot N$.
 \item Let $F$ be an arbitrary set of $2\delta N$ additional edges such that $P = M_i \cup M_i' \cup F$ is a perfect matching in $G^{RS}$.
 \item Consider distribution $\lambda$ constructed using RS graph $G^{RS} \setminus (M_i \cup M_i')$.
 \item For every $G = (V,E) \sim \lambda$, let $P_G =  M_i \cup M_i' \cup (F \setminus E)$ (to avoid multi-edges) and add $P_G$ to $G$ to obtain the input graph $G^+$. 
\end{enumerate}
The edges $P \cup E_1$ are given to Alice and the edges $P \cup E_2$ are given to Bob (recall that $E_1$ and $E_2$ are defined in distribution $\lambda$).
\end{minipage}
}
\caption{Hard input distribution $\lambda^+$. \label{fig:lambda+}}\end{figure}

We are now ready to prove our main lower bound theorem:
\begin{theorem} \label{thm:CC_lower}
For any $\epsilon > 0$, every deterministic $(\frac{2}{3} + \epsilon)$-approximation one-way communication protocol with constant distributional error over $\lambda^+$ for \textsf{MBM} requires a message of size $n^{1+\Omega(\frac{1}{\log \log n})}$, where $n$ is the number of vertices in the input graph.
\end{theorem}

\begin{proof}
Let $\gamma^+$ be a deterministic $(\frac{2}{3} + \epsilon)$-approximation protocol that solves distribution $\lambda^+$ with constant distributional error. Given $\gamma^+$, we will now define a protocol $\gamma$ that solves distribution $\lambda$ with the same communication cost, same error, and approximation ratio strictly better than $\frac{2}{3}$. Invoking Theorem~\ref{thm:goel-hardness} then proves our result.

The protocol $\gamma$ is easy to obtain: 
Observe that $P$ in distribution $\lambda^+$ is the same for every sampled input graph $G^+ \sim \lambda^+$. Hence, in protocol $\gamma$, Alice and Bob first make sure that the edges $P$ are included in their inputs. This is achieved by Alice adding the edges $P \setminus E_1 = P_G$ to her input, and Bob adding the edges $P$ to his input. In doing so, Alice and Bob's input is equivalently distributed to choosing an input graph $G^+$ from $\lambda^+$. Alice and Bob can, therefore, run protocol $\gamma^+$ which produces an output matching $M_{\text{out}}^+$. Bob then outputs the largest matching $M_{\text{out}}$ among the edges $M^*_X \cup M^*_Y \cup (M_{\text{out}}^+  \setminus P_G)$ as the output of the protocol $\gamma$.

Next, we will argue that $|M_{\text{out}}| \ge |M_{\text{out}}^+| - |F| = |M_{\text{out}}^+| - 2 \delta N$. We can construct a matching $\tilde{M}$ of this size as follows: First, add every edge $e  \in M_{\text{out}}^+$ that is not contained in $P$ to $\tilde{M}$. Second, for every edge $e \in M_{\text{out}}^+ \cap (M_i \cup M_i')$, we insert the incident edge to $e$ that is contained in $M_X^* \cup M_Y^*$ into $\tilde{M}$ (notice that these incident edges always exist except for edges from the special induced matching). This implies that $|M_{\text{out}}| \ge |\tilde{M}| \ge |M_{\text{out}}^+| - |F|$.

Recall that $\mu(G) \ge \frac{3}{2}N$ and, since $G$ is a subgraph of $G^+$, $\mu(G^+) \ge \mu(G)$. This implies that $N \le \frac{2}{3} \mu (G^+)$. Since $\gamma^+$ is a $(\frac{2}{3}+\epsilon)$-approximation protocol, we have 
$|M_{\text{out}}^+| \ge (\frac{2}{3}+\epsilon)\mu(G^+)$, and thus:
\[|M_{\text{out}}| \ge |M_{\text{out}}^+| - 2 \delta N \ge (\frac{2}{3}+\epsilon)\mu(G^+)  - 2 \delta \frac{2}{3} \mu(G^+) = (\frac{2}{3} + \epsilon - \frac{4}{3} \delta)\mu(G^+) \ .\]
Hence, setting $\delta < \frac{3}{4} \epsilon$ in distribution $\lambda$ yields a protocol with approximation ratio strictly above $\frac{2}{3}$. This, however, implies that $\gamma$ requires a message of length $n^{1+\Omega(\frac{1}{\log \log n})}$ (Theorem~\ref{thm:goel-hardness}), and since the message sent in $\gamma$ and $\gamma^+$ is equivalent, the result follows.
\end{proof}
Applying Yao's Lemma~and the usual connection between streaming algorithms and one-way communication protocols, we obtain our main lower bound result:

\begin{manualtheorem}{\ref{thm:lb}} 
 For any $\epsilon > 0$, every (possibly randomised) two-pass streaming algorithm for \textsf{MBM} with approximation ratio $\frac{2}{3} + \epsilon$ that solely computes a \textsc{Greedy} matching in the first pass requires $n^{1+\Omega(\frac{1}{\log \log n})}$ space, where $n$ is the number of vertices in the graph.
\end{manualtheorem}

\section{Algorithm} \label{sec:ub}

In this section, we combine the subsampling approach as used by Konrad \cite{k18} and the semi-matching approach as used by Esfandiari et al. \cite{ehm16} and Kale and Tirodkar \cite{kt17} in order to find many disjoint $3$-augmenting paths, see Algorithm~\ref{alg:MainAlgo}.

The input to Algorithm~\ref{alg:MainAlgo} is a stream of edges $\pi$ of a bipartite graph $G=(A, B, E)$, a maximal matching $M$ in $G$ (e.g., computed in a first pass by \textsc{Greedy}), a sampling probability $p$, and an integral degree bound $d$. First, each edge of $M$ is included in $M'$ with probability $p$. Then, while processing the stream, degree-$d$-bounded semi-matchings $S_L$ and $S_R$ are computed using the algorithm $\textsc{Greedy}_d$ (see Algorithm~\ref{alg:greedyd} in Section~\ref{sec:intro}). The algorithm then returns a largest subset of vertex-disjoint $3$-augmenting paths $\mathcal{Q}$. We can thus obtain a matching of size $|M| + |\mathcal{Q}|$.

\begin{figure}[t]
\centering
\begin{minipage}{0.95\textwidth}
\begin{algorithm}[H]
\caption{Finding Augmenting Paths.}
\label{alg:MainAlgo}
\textbf{Input:} A stream of edges $\pi$ of a bipartite graph $G =(A,B,E)$, a maximal matching $M$ in $G$, $p \in (0, 1]$ and $d \in \mathbb{N}^+$.
\begin{algorithmic}[1]
\State Let $M' \subseteq M$ be a random subset such that $\forall e \in M$, $\Pr[e \in M'] = p$
\State Let $G_L' = G[A(M') \cup \overline{B(M)}]$ and $G_R' = G[\overline{A(M)} \cup B(M')]$ 
\State Denote by $\pi_{G_L'}$ ($\pi_{G_R'}$) the substream of $\pi$ of edges of $G_L'$ ($G_R'$, respectively)
\State $S_L \leftarrow \textsc{Greedy}_d(\pi_{G_L'})$ such that $\deg_{S_L}(a) \le 1$, for every $a \in A(M')$, and $\deg_{S_L}(b) \le d$, for every $b \in \overline{B(M)}$
\State $S_R \leftarrow \textsc{Greedy}_d(\pi_{G_R'})$ such that $\deg_{S_R}(b) \le 1$, for every $b \in B(M')$, and $\deg_{S_R}(a) \le d$, for every $a \in \overline{A(M)}$
\State $ \mathcal{P} \leftarrow \{ ab', ab, a'b : ab' \in S_L, ab \in M', a'b \in S_R  \}$
\State \textbf{return} A largest subset $\mathcal{Q} \subseteq \mathcal{P}$ of vertex-disjoint paths.
\end{algorithmic}
\end{algorithm}
\end{minipage}
\end{figure}

\subsection{Analysis of Algorithm~\ref{alg:MainAlgo}}

The main task in analysing Algorithm~\ref{alg:MainAlgo} is to bound the sizes of $S_L$ and $S_R$ from below. A bound that holds in expectation for the case $d=1$ was previously proved by Konrad et al. \cite{kmm12}, and a high probability result (for $d=1$) was later obtained by Konrad \cite{k18}. We also first give a bound that holds in expectation (Lemma~\ref{lem:main_expec}), which is achieved by extending the original proof by Konrad et al. \cite{kmm12}. Our extension, however, is non-trivial as it requires a very different progress measure. Then, following Konrad \cite{k18}, we obtain a high probability version in Lemma~\ref{lem:main_highp}.

We also remark that Lemmas~\ref{lem:main_expec} and \ref{lem:main_highp} are stated in a more general context, however, it is not hard to see that they capture the situation of the computations of $S_L$ and $S_R$ in subgraphs $G_L'$ and $G_R'$, respectively.

\begin{lemma} \label{lem:main_expec}
Let $G = (A, B, E)$ be a bipartite graph, $\pi$ an arbitrarily ordered stream of its edges, $p \in (0,1]$, and $d \in \mathbb{N}^+$. Let $A' \subseteq A$ be a random subset such that $\forall a \in A$, $\Pr[a \in A'] = p$, and let $d$ be the degree bound of the $B$ vertices. Let $H = G[A' \cup B]$ and denote by $\pi_H$ the substream of $\pi$ consisting of the edges in $H$. Then,
\[\mathbb{E}_{A'}[|\textsc{Greedy}_d(\pi_H)|] \geq \frac{d }{d + p} \cdot p \cdot \mu(G) \ .\]
\end{lemma}

\begin{proof}
Let $M^*$ be a fixed maximum matching in $G$ and let $M^*_H := \{ ab \in M^* : a \in A' \}$ be the subset of edges incident to $A'$. The goal is to bound the expected number of edges in $M^*_H$ blocked by the semi-matching returned by $\textsc{Greedy}_d(\pi_H)$.

\subparagraph*{\emph{Game Setup.}} Consider the following game: On selection of an edge by $\textsc{Greedy}_d(\pi_H)$, the edge \textit{attacks} the (at most two) incident edges of $M^*_H$ and deals damage to them. Initially, the damage of every edge in $M^*_H$ is $0$, and the maximum damage of each such edge is $1$. A damage below $1$ means that the edge could still be selected by the algorithm. A damage equal to $1$ implies that the edge can no longer be selected.

Denote by $S_i$ the first $i$ edges selected by $\textsc{Greedy}_d(\pi_H)$ and let $ab$ be the $(i + 1)^{\text{th}}$ edge selected.  The way damage is dealt is as follows:
\begin{itemize}
\item If there is an edge $a'b \in M^*_H$ such that $a' \notin A(S_{i+1})$ then attack edge $a'b$ by adding $\frac{1}{d}$ damage to it; 
\item If there is an edge $ab' \in M^*_H$ then attack edge $ab'$ by adding $1 - \frac{\deg_{S_i}(b')}{d}$ damage to it, maxing out the damage to $1$.
\end{itemize}

Observe that the maximum damage which any edge selected by $\textsc{Greedy}_d(\pi_H)$ can inflict is $1+\frac{1}{d}$ (applying both cases to the two incident optimal edges). Furthermore, observe that the maximum damage which every edge in $M^*_H$ can receive is $1$, and, indeed, at the end of the algorithm, every edge in $M^*_H$ has damage $1$.

\subparagraph*{\emph{Applying Wald's Equation.}} Denote by $s$ the cardinality of the semi-matching computed by $\textsc{Greedy}_d(\pi_H)$ and let $X_1, X_2, \dots, X_s$ be the sequence of edges selected. Define the random variable $Y_i$ to be the damage dealt by edge $X_i$. Let $T$ be the smallest $i$ such that $\sum_{j=1}^i Y_j = |M^*_H|$ holds. Observe that $T$ is a random stopping time. To apply the version of Wald's Equation presented in Lemma~\ref{lem:walds_eq}, we need to show that $\mathbb{E}[T]$ is finite and find a value $\tau$ such that, for all $i \leq T$, $\mathbb{E}[Y_i] \leq \tau$ holds:

The expected stopping time $\mathbb{E}[T]$ is finite since $T \le s$ always holds by the end of the algorithm, i.e., the total damage dealt is $|M^*_H|$. Finding $\tau$ is less obvious. By definition, the damage $Y_i$ dealt by any edge $X_i$ is either $0, \frac{1}{d}, \dots, 1$ or $1 + \frac{1}{d}$. Hence, we obtain the following:
\begin{align*}
\mathbb{E}[Y_i] & \le \Pr \left[ Y_i \le 1 \right] \cdot 1 + \underbrace{\Pr \left[ Y_i = 1 + \frac{1}{d} \right]}_q \cdot (1+\frac{1}{d}) = (1-q) \cdot 1  + q \cdot (1+\frac{1}{d}) =  1 + \frac{q}{d} \ .
\end{align*}
It remains to bound $ \Pr [ Y_i = 1 + \frac{1}{d}](=q)$. Let $X_i = ab$.  Then, by definition of the game, the event $Y_i = 1 + \frac{1}{d}$ only happens if there exists an edge $a'b \in M^*_H$ such that $a' \notin A(S_i)$. In this case, $ab$ inflicts a damage of $1$ on edge $a'b$. However, observe that since $a' \notin A(S_i)$, the random choice as to whether $a' \in A'$ and thus whether $a'b \in M^*_H$ had not needed to occur yet (principle of deferred decision). Hence, we obtain: \[\Pr[Y_i = 1 + \frac{1}{d}] \le \Pr[a' \in A'] = p \ .\]

Having shown that $\mathbb{E}[T]$ is finite and $\mathbb{E}[Y_i] \leq 1+\frac{p}{d}$ for all $i \leq T$, we can apply Wald's Equation (Lemma~\ref{lem:walds_eq}) and we obtain $\mathbb{E}[\sum_{j=1}^T Y_i] \leq (1 + \frac{p}{d}) \mathbb{E}[T]$. Finally, since $\mathbb{E}[\sum_{j=1}^T Y_i] = \mathbb{E}[|M^*_H|] = p \cdot \mu(G)$ and $T \leq s = |\textsc{Greedy}_d(\pi_H)|$, it follows that 
\[\mathbb{E}[\sum_{j=1}^T Y_i] =  p \cdot \mu(G) \leq (1 + \frac{p}{d}) \cdot \mathbb{E}[T] \leq  (1 + \frac{p}{d}) \cdot \mathbb{E}[|\textsc{Greedy}_d(\pi_H)] \ ,\]
which implies the result.
\end{proof}

Next, we follow the approach by Konrad \cite{k18} to strengthen Lemma~\ref{lem:main_expec} and obtain the following high probability result (see Appendix~\ref{app:highprob} for the proof):

\begin{restatable}{lemma}{mainhighp} \label{lem:main_highp}
Let $G = (A, B, E)$ be a bipartite graph, $\pi$ be any arbitrary stream of its edges, $p \in (0,1]$ and $d \in \mathbb{N}^+$. Let $A' \subseteq A$ be a random subset such that $\forall a \in A$, $\Pr[a \in A'] = p$, let $d$ be the degree bound of the $B$ vertices and let $H = G[A' \cup B]$. Then, the following holds with probability at least $1 - 2\mu(G)^{-18}$:
\[|\textsc{Greedy}_d(\pi_H)| \geq \frac{d }{d + p} \cdot p \cdot \mu(G) - o(\mu(G)).\]
\end{restatable}

Equipped with Lemma~\ref{lem:main_highp}, we are now ready to bound the number of augmenting paths found by Algorithm~\ref{alg:MainAlgo}. 

\begin{restatable}{lemma}{mainalgoguarantee}  \label{lem:main_algo_guarantee}
Suppose that $|M| = (\frac{1}{2}+\epsilon)\mu(G)$. Then, 
with probability $1 - \mu(G)^{-16}$, the number of vertex-disjoint $3$-augmenting paths $|\mathcal{Q}|$ found by
Algorithm~\ref{alg:MainAlgo} is at least:
\[|\mathcal{Q}| \ge (\frac{1 - 2\epsilon}{d + p} - \frac{1 + 2\epsilon}{2d}) \cdot p \cdot \mu(G) - o(\mu(G)) \ .\] 
\end{restatable}

\begin{proof}
Let $M^*$ be a fixed maximum matching in $G$.
In this proof, we will refer to the quantities used by Algorithm~\ref{alg:MainAlgo}. First, using a Chernoff bound for independent Poisson trials, we see that $|M'| = p \cdot |M| \pm O(\sqrt{|M| \ln |M|})$ with probability at least $1 - |M|^{-C}$ for an arbitrarily large constant $C$.

Consider the subgraphs $G_L = G[A(M) \cup \overline{B(M)}]$ and $G_R = G[\overline{A(M)} \cup B(M)]$. $M \oplus M^*$ contains $(\frac{1}{2} - \epsilon)\mu(G)$ vertex-disjoint augmenting paths where each path starts and ends with an edge in $G_L \cup G_R$. This implies that
\begin{align}\label{eq:sum_musubG}
\mu(G_L) + \mu(G_R) & \ge  2(\frac{1}{2} - \epsilon)\mu(G) = (1 - 2 \epsilon) \mu(G) \  . 
\end{align}

Following Konrad \cite{k18}, we will next argue the following:
\begin{align}
|\mathcal{P}| \ge |S_L| + |S_R| - |M'| \ . \label{eqn:p}
\end{align}
Observe that there are $|M'|-|S_L|$ vertices of $|M'|$ that are not incident to an edge in $S_L$, and similarly, $|M'|-|S_R|$ vertices of $|M'|$ that are not incident to an edge in $S_R$. Hence, there are at least $|M'| - (|M'|-|S_L|) - (|M'|-|S_R|) = |S_L| + |S_R| - |M'|$ edges of $|M'|$ that are incident to both an edge from $S_L$ and $S_R$. We thus obtain that there are at least $|\mathcal{P}| \ge |S_L| + |S_R| - |M'|$ $3$-augmenting paths.

Next, Esfandiari et al. (Lemma 6 in \cite{ehm16}) consider a similar structure to $\mathcal{P}$ and argue that there is at least a $d$-fraction of augmenting paths in $\mathcal{P}$ that are vertex-disjoint, and, hence, 
\begin{align}
|\mathcal{Q}| \ge \frac{1}{d} |\mathcal{P}| \ . \label{eqn:q}
\end{align}

Using Lemma~\ref{lem:main_highp} and  Inequalities~\ref{eq:sum_musubG}, \ref{eqn:p}, and \ref{eqn:q}, we obtain:
\begin{align*}
|\mathcal{Q}| &\geq \frac{1}{d}(|S_L| + |S_R| - |M'|) \\
&\geq \frac{1}{d}(\frac{d }{d + p} \cdot p \cdot (1 - 2\epsilon)\mu(G) - o(\mu(G)) - p \cdot (\frac{1}{2} + \epsilon)\mu(G) - O(\sqrt{\mu(G) \ln \mu(G)})) \\
&= (\frac{1 - 2\epsilon}{d + p} - \frac{1 + 2\epsilon}{2d}) \cdot p \cdot \mu(G) - o(\mu(G)).
\end{align*}
Using the union bound, the error of the algorithm is bounded by $|M|^{-C} + 2\mu(G)^{-18} \leq \mu(G)^{-16}$.
\end{proof}

We are now ready to state our main algorithmic result:

\begin{manualtheorem}{\ref{thm:ub}}
For every $p \in (0,1]$ and every integral $d \ge 1$, there is a two-pass semi-streaming algorithm for \textsf{MBM} with approximation factor 
\begin{align*}
  \begin{cases}
   \frac{1}{2} + (\frac{1}{d+p} - \frac{1}{2d}) \cdot p - o(1), & \mbox{ if } p \le d(\sqrt{2}-1) \\
   \frac{1}{2} + \frac{d-p}{6d+2p} - o(1), & \ \mbox{ otherwise} \ ,
  \end{cases}
 \end{align*}
that succeeds with high probability (in $\mu(G)$, where $G$ is the input graph). The settings ($d=1, p=\sqrt{2}-1$) and ($d=2, p=2(\sqrt{2}-1)$) maximize the approximation factor to $2-\sqrt{2}- o(1)$.
\end{manualtheorem}

\begin{proof}
Let $M$ be a maximal matching such that $|M| = (\frac{1}{2}+\epsilon)\mu(G)$, for some $0 \le \epsilon \le \frac{1}{2}$ and some bipartite graph $G=(A, B, E)$ with a stream $\pi$ of its edges. Let $\mathcal{Q}$ be the disjoint augmenting paths found by Algorithm~\ref{alg:MainAlgo} on input $\pi, M, p$ and $d$. Then, augmenting $M$ with $\mathcal{Q}$ yields a matching of size $|M| + |\mathcal{Q}|$. By Lemma~\ref{lem:main_algo_guarantee}, the following inequality holds with high probability:
\begin{align}
|M| + |\mathcal{Q}| \geq (\frac{1}{2} + \epsilon)\mu(G) + (\frac{1 - 2\epsilon}{d + p} - \frac{1 + 2\epsilon}{2d}) \cdot p \cdot \mu(G) - o(\mu(G)). \label{eqn:392}
\end{align}

We distinguish two cases:
\begin{enumerate}
 \item If $p \le d(\sqrt{2}-1)$ then $\epsilon = 0$ minimizes the RHS of Inequality~\ref{eqn:392}, and we obtain the claimed bound by plugging the value $\epsilon = 0$ into the inequality.
 \item If $p \ge d(\sqrt{2}-1)$ (only possible if $d \in \{1, 2 \}$) then $\epsilon = \frac{d-p}{6d+2p}$ minimizes the RHS of Inequality~\ref{eqn:392}, and we obtain the claimed bound by plugging the value $\epsilon = \frac{d-p}{6d+2p}$ into the inequality.
\end{enumerate}
It can be seen that, for a fixed $d$, the maximum is obtained if $p = \min \{d \sqrt{2} - d, 1 \}$, and the values $d \in \{1, 2\}$ yield the claimed bound of $2-\sqrt{2} - o(1)$ (see Figure~\ref{fig:meta_alg} in Section~\ref{sec:intro}).
\end{proof}

\subsection{Optimality of the Analysis} \label{sec:instance-optimality}
We will show now that our analysis of Algorithm~\ref{alg:MainAlgo} is best possible. To this end, we define a worst-case input graph $G$ in Figure~\ref{fig:hard-instance}, and prove in Theorem~\ref{thm:opt_analysis} that Algorithm~\ref{alg:MainAlgo} does not perform better on $G$ than predicted by our analysis. See Figure~\ref{fig:opt_graph} for an illustration.

\begin{figure}[t]
\centering
\fbox{
\begin{minipage}{0.95\textwidth}
\begin{enumerate}
\item Let 
$A_{\text{in}} = \{a_{\text{in}}^1, a_{\text{in}}^2, \dots, a_{\text{in}}^N \}$, 
$A_{\text{out}} = \{a_{\text{out}}^1, \dots, a_{\text{out}}^N \}$, 
$B_{\text{in}} = \{b_{\text{in}}^1, \dots, b_{\text{in}}^N \}$, and
$B_{\text{out}} = \{b_{\text{out}}^1, \dots, b_{\text{out}}^N \}$ be sets of vertices, for some integer $N$.
\item Let $M = \{ a_{\text{in}}^i b_{\text{in}}^i \ : \ 1 \le i \le N \}$ be a perfect matching between $A_{\text{in}}$ and $B_{\text{in}}$. Let $G_L=(A_{\text{in}}, B_{\text{out}}, E_L)$ be a semi-complete graph
such that $a_{\text{in}}^i b_{\text{out}}^j \in E_L \Leftrightarrow i \geq j$, and let $G_R=(A_{\text{out}}, B_{\text{in}}, E_R)$ be a semi-complete graph such that $a_{\text{out}}^i b_{\text{in}}^j \in E_R \Leftrightarrow i \geq j$. 
\item Our bipartite hard instance graph is defined as $G = (A_{\text{in}} \cup A_{\text{out}}, B_{\text{in}} \cup B_{\text{out}}, M \cup E_L \cup E_R)$ and has $n = 4N$ vertices. 
\item Finally, let $\pi$ be a stream of its edges where the edges of $M$
arrive first followed by the edges $E_L$ and $E_R$. The edges in $E_L$ are ordered so that $a_{\text{in}}^i b_{\text{out}}^j$ arrives before $a_{\text{in}}^{i'} b_{\text{out}}^{j'}$ only if $i > i'$, or $i = i'$ and $j < j'$. Similarly, the edges in $E_R$ are ordered so that $a_{\text{out}}^i b_{\text{in}}^j$ arrives before $a_{\text{out}}^{i'} b_{\text{in}}^{j'}$ only if $i > i'$, or $i = i'$ and $j < j'$.
\end{enumerate}
\end{minipage}
}
\caption{\label{fig:hard-instance}Hard input instance $G$ for Algorithm~\ref{alg:MainAlgo}.}
\end{figure}

\begin{figure}[t]
\centering
\includegraphics[width=0.9\textwidth]{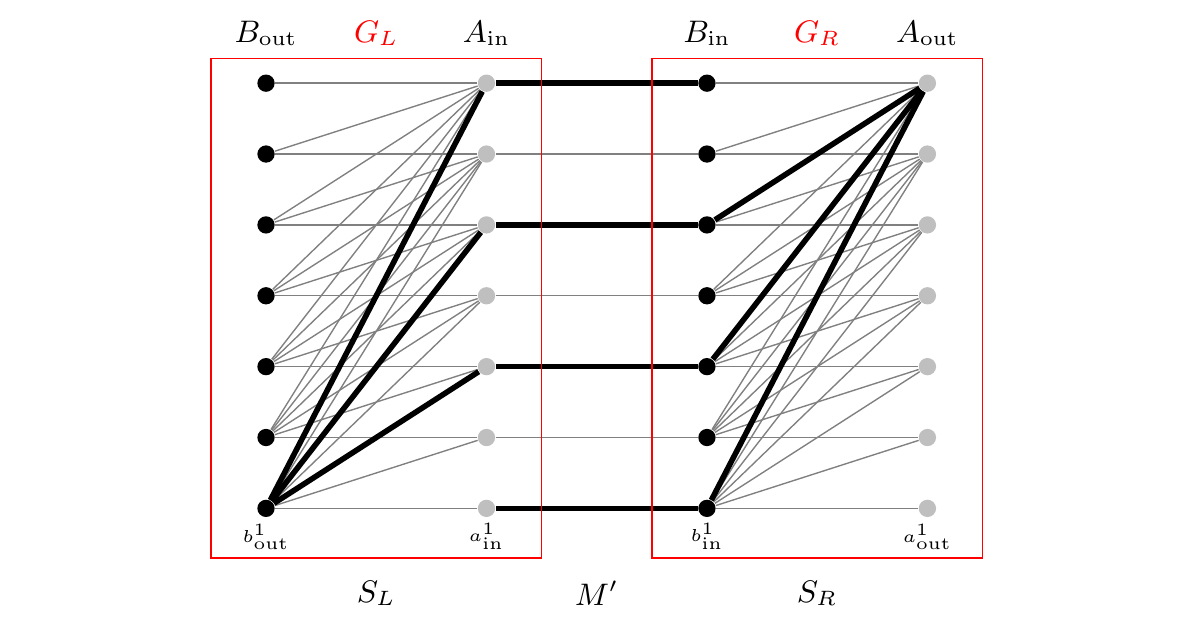}
\caption{Algorithm~\ref{alg:MainAlgo} on a hard input instance with $N=7$, $d=3$ and $p = 0.5$.} \label{fig:opt_graph}
\end{figure}

Observe that $M$ is a maximal matching in $G$, and if we run \textsc{Greedy} in the first pass on $\pi$ then $M$ would be returned. Let $M^*_L = \{ a_{\text{in}}^i b_{\text{out}}^i \ : \ 1 \le i \le N \}$ and $M^*_R = \{ a_{\text{out}}^i b_{\text{in}}^i \ : \ 1 \le i \le N \}$. Then, $M^*_L$ is a perfect matching in $G_L$, $M^*_R$ is a perfect matching in $G_R$, and $M^*_L \cup M^*_R$ is a perfect matching in $G$.

\begin{theorem} \label{thm:opt_analysis}
 Algorithm~\ref{alg:MainAlgo} with parameters $d \ge 1$ and $0 < p \le 1$ on input $G$ received via stream $\pi$ and maximal matching $M$ finds at most 
\[\left( (\frac{1}{d+p} - \frac{1}{2d}) \cdot p + o(1) \right) \mu(G) \]
augmenting paths with high probability. This renders our analysis of Algorithm~\ref{alg:MainAlgo} best possible when $p \le d \sqrt{2} - d$.
\end{theorem}

\begin{proof}

In this proof, we will refer to the quantities used by Algorithm~\ref{alg:MainAlgo}, that is, $M'$ (the edges of $M$ subsampled with probability $p$), $S_L$ and $S_R$. 

We will use the following claim in our proof: 
\begin{claim}\label{claim:opt}
With high probability, for every pair $i,j \in [N]$ with $i \le j$, we have 
\begin{align*}
|\{ a_{\text{in}}^k b_{\text{in}}^k \in M' \ | \ i \le k \le  j\} | & = p \cdot (j - i) \pm o(N) \ .
\end{align*}
\end{claim}
\begin{proof}
This claim is easy to prove. Indeed, for any fixed $i,j \in [N]$ with $i \le j$, the statement above follows directly from the Chernoff bound. Using the union bound over all pairs $i,j \in [N]$, the claim follows.
\end{proof}
\noindent From now on, we condition on the event that the statement in Claim~\ref{claim:opt} holds.

Let $A_{\text{in}}' = A(M')$ and let $B_{\text{in}}' = B(M')$. We will first argue that, for two different vertices $a_{\text{in}}^i, a_{\text{in}}^j \in A_{\text{in}}'$ with $i < j$, if $a_{\text{in}}^i \in A(S_L)$ then $a_{\text{in}}^j \in A(S_L)$ also holds. 
Indeed, suppose that this was not the case. Let $b_{\text{out}}^k$ be the partner of $a_{\text{in}}^i$ in $S_L$. Observe that the edges $a_{\text{in}}^i b_{\text{out}}^k, a_{\text{in}}^j b_{\text{out}}^k \in E_L$, and, in particular, the edge $a_{\text{in}}^j b_{\text{out}}^k$ arrives before the edge $a_{\text{in}}^i b_{\text{out}}^k$ in $\pi$. Hence, edge $a_{\text{in}}^j b_{\text{out}}^k$ would have been selected, a contradiction. A similar argument holds for vertices $b_{\text{out}}^i, b_{\text{out}}^j \in B_\text{out}$ with $i > j$; if $\text{deg}_{S_L}(b_\text{out}^i) \geq 1$ then $\text{deg}_{S_L}(b_\text{out}^j) = d$.

Let $i_{\text{min}}$ be the smallest index such that $a_{\text{in}}^{i_{\text{min}}} \in A(S_L)$. We will now argue that $i_{\text{min}} \ge \frac{pN}{p+d} - o(N)$.
Observe that the vertices $A_{\text{in}}'$ are matched in order from the largest to smallest index, and each matched vertex in $A_{\text{in}}'$ is matched only once. The vertices in $B_{\text{out}}$ are matched from the smallest to largest index, and each matched vertex is matched $d$ times (except possibly the last such matched vertex). Consider the last edge $a_{\text{in}}^{i_{\text{min}}} b_{\text{out}}^q$ inserted into $S_L$. Then, $q \le i_{\text{min}}$, and, thus, $|B(S_L)| \le i_{\text{min}}$.
By Claim~\ref{claim:opt} (applied with $j = N$), we have $|A(S_L)| \ge p \cdot (N - i_{\text{min}}) - o(N)$ with high probability. 
Since $|A(S_L)|$ is matched to $B(S_L)$ in $S_L$, and each $B$-vertex is matched at most $d$ times, we obtain $|A(S_L)| \le d \cdot |B(S_L)|$, and, hence:
\[
  p \cdot (N - i_{\text{min}}) - o(N) \le |A(S_L)| \le  d \cdot |B(S_L)| \le d \cdot i_{\text{min}} \ ,
\]
which implies $i_{\text{min}} \ge \frac{pN}{p+d} - o(N)$.

Let $i_{\text{max}}$ be the largest index such that $b_{\text{in}}^{i_{\text{max}}} \in B(S_R)$. Using a similar argument as above, we see that 
$i_{\text{max}} \le \frac{dN}{p+d} + o(N)$. 

Let $M'' = \{  a_{\text{in}}^i b_{\text{in}}^i \in M' : i_{\text{min}} \le i \le i_{\text{max}} \}$ be the subset of augmentable edges, i.e., edges for which there exists a left wing in $S_L$ and a right wing in $S_R$. Then, by Claim~\ref{claim:opt}, we have
\[|M''| \le p \cdot \left( i_{\text{max}} - i_{\text{min}} \right) + o(N) \leq  \frac{p(d-p)N}{p+d} + o(N) \ . \] 
All but constantly many vertices in $A(M'')$ share the same neighbour in $S_L$ with $d-1$ other vertices of $A(M'')$. Hence, at most a $d$-fraction (plus up to the constantly many exceptions, which disappear in the $o(N)$ term) of $M''$ can be augmented simultaneously. Using $N = \frac{1}{2} \mu(G)$, we obtain the following bound on the number of edges that can be augmented simultaneously: 
\[\frac{1}{d} |M''| \le \frac{1}{d} \left( \frac{p(d-p)N}{p+d} + o(N) \right) =  \left( (\frac{1}{d+p} - \frac{1}{2d}) \cdot p + o(1) \right) \mu(G) \ . \]
\end{proof}

\section{Conclusion} \label{sec:conclusion}
In this paper, we studied the class of two-pass semi-streaming algorithms for \textsf{MBM} that solely compute a \textsc{Greedy} matching in the first pass. We showed that algorithms of this class cannot have an approximation ratio of $\frac{2}{3}+\epsilon$, for any $\epsilon > 0$. We also combined the two dominant techniques that have previously been used for designing such algorithms and discovered another algorithm that matches the state-of-the-art approximation factor of $2-\sqrt{2} \approx 0.58578$.

We conclude with two open problems. First, we are particularly interested in whether there exists a one-pass semi-streaming algorithm that is able to augment a maximal matching so as to yield an approximation ratio above $2-\sqrt{2}$. Second, is there a two-pass semi-streaming algorithm for \textsf{MBM} that improves on the approximation factor of $2-\sqrt{2}$ and operates differently in the first pass to the class of algorithms considered in this paper?
 
\pagebreak
\printbibliography[heading=bibintoc]
\pagebreak
\appendix

\section{Strengthening Lemma~\ref{lem:main_expec}} \label{app:highprob}

Following \cite{k18}, we use tail inequalities for martingales to strengthen Lemma~\ref{lem:main_expec} and give a high probability result. The proof of Lemma~\ref{lem:main_highp} uses the \textit{Azuma-Hoeffding's Inequality} \cite[Theorem~12.4]{mitzenmacher2005probability}:

\begin{lemma}[Azuma-Hoeffding's Inequality] \label{lem:Azuma-Hoeffding}
Let $Z_0, Z_1, ..., Z_n$ be a martingale such that $\forall k \geq 0$, $|Z_{k+1} - Z_k| \leq c_k$. Then, $\forall t \geq 0$ and any $\lambda > 0$,
\[Pr\left[|Z_t - Z_0| \geq \lambda \right] \leq 2\exp \left(\frac{-\lambda^2}{2\sum_{k=0}^{t-1}c_k^2} \right).\]
\end{lemma}

\mainhighp*

\begin{proof}
Let $X_1, X_2, \dots, X_s$ be the sequence of random variables representing the edges selected by $\textsc{Greedy}_d(\pi_H)$ with the source of randomness from the choice of $A'$. Define $Y := |\textsc{Greedy}_d(\pi_H)|$. Then, we define the random variables $Z_i := \mathbb{E}[Y | X_1, ..., X_i]$ for all $i = 0, \dots, s$ to be the corresponding Doob Martingale, and let $Z_i = Z_{i-1}$, for every $i > s$. Notice that $Z_s = Y$ and $ Z_0 = \mathbb{E}[Y] \geq \frac{d }{d + p} \cdot p \cdot \mu(G)$ by Lemma~\ref{lem:main_expec}. Now, we will show that any deviation of $Y$ from its expectation, $|Z_s - Z_0|$, is small with high probability.

To that end, we first need to bound $|Z_{i+1} - Z_i|$ for all $i \geq 0$. Notice that $|Z_{i+1} - Z_i| = 0$ for all $i \geq s$. Next, we will argue that $|Z_{i+1} - Z_i| \leq 1$ for all $i < s$. Indeed, for any fixed first $i$ edges added to the semi-matching, any two different choices for $X_{i+1}$ yield two potentially different semi-matchings $S_1, S_2$, respectively, such that $S_1 \oplus S_2$ consists of at most one alternating path. Hence, the two semi-matchings differ by at most one edge, which proves the claim.

Then, we have that $s = Y \leq d \cdot \mu(H) \leq d \cdot \mu(G)$ and it follows that $|Z_{i+1} - Z_i| \leq 1$ for all $i \leq d \cdot \mu(G)$ and $|Z_{i+1} - Z_i| = 0$ for all $i > d \cdot \mu(G)$. Finally, by applying Azuma-Hoeffding's Inequality (see Lemma~\ref{lem:Azuma-Hoeffding}), we finalise the proof:
\begin{align*}
\Pr \left[|Z_s - Z_0| \geq 6\sqrt{d\mu(G) \ln \mu(G)} \right] \leq 2\mu(G)^{-18} \ , 
\end{align*}
where $|Z_s - Z_0| = |Y - \mathbb{E}[Y]|$.
\end{proof}
\end{document}